\documentclass{llncs}

\usepackage[dvipsnames]{xcolor}
\usepackage{complexity}
\usepackage{tikz}
\usetikzlibrary{positioning,arrows.meta,calc,fit}
\usepackage{amsmath}
\usepackage{algorithm}
\usepackage{algpseudocode}

\newclass{\LBAACCEPTANCE}{LBA\text{--}ACCEPTANCE}
\newclass{\SEQCONTROL}{SEQ\text{--}CONTROL}

\title{A P Systems Variant for Reasoning about Sequential
  Controllability of Boolean Networks%
  \thanks{This is a revised and extended version of \cite{AlhazovFI2022}.}
}

\author{Artiom Alhazov\inst{1} \and
  Vincent Ferrari-Dominguez\inst{2} \and
  Rudolf Freund\inst{3} \and
  Nicolas~Glade\inst{4} \and
  Sergiu Ivanov\inst{5}\thanks{Corresponding author.}
}

\institute{
Vladimir Andrunachievici Institute of Mathematics and Computer Science,\\
The State University of Moldova\\
Academiei 5, Chi\c sin\u au, MD-2028, Moldova\\
{\sf{artiom@math.md}}\\[0.7em]
\and
\'{E}cole Normale Sup\'{e}rieure -- PSL University, CNRS UMR 8553,\\
45, rue d'Ulm, F-75230 Parix Cedex 3, France
\\
{\sf{vincent.ferrari@ens.psl.eu}}\\[0.7em]
\and
Faculty of Informatics, TU Wien\\
Favoritenstra\ss{}e 9--11, 1040 Wien, Austria\\
{\sf{rudi@emcc.at}}\\[0.7em]
\and
Universit\'{e} Grenoble Alpes, CNRS UMR5525, CHU Grenoble Alpes, Grenoble INP\\ TIMC-IMAG, F-38000 Grenoble, France\\
{\sf{nicolas.glade@univ-grenoble-alpes.fr}}\\[0.7em]
\and
Universit\'{e} Paris-Saclay, Univ. \'{E}vry, IBISC,\\
23 boulevard de France, 91034 \'{E}vry, France\\
{\sf{sergiu.ivanov@ibisc.univ-evry.fr}}\\[0.7em]
}

\newcommand{\dottimes}{\mathop{\dot{\times}}}

\begin{document}
\maketitle

\begin{abstract}
  A Boolean network is a discrete dynamical system operating on
  vectors of Boolean variables.  The action of a Boolean network can
  be conveniently expressed as a system of Boolean update functions,
  computing the new values for each component of the Boolean vector as
  a function of the other components.  Boolean networks are widely
  used in modelling biological systems that can be seen as consisting
  of entities which can be activated or deactivated, expressed or
  inhibited, on or off.  P systems on the other hand are classically
  introduced as a model of hierarchical multiset rewriting.  However,
  over the years the community has proposed a wide range of P system
  variants including diverse ingredients suited for various needs.
  In this work, we propose a new variant---Boolean
  P systems---specifically designed for reasoning about sequential
  controllability of Boolean networks, and use it to first establish
  a crisp formalization of the problem, and then to prove that the
  problem of sequential controllability is \PSPACE-complete.
  We further claim that Boolean P systems are a demonstration of how
  P systems can be used to construct ad hoc formalisms,
  custom-tailored for reasoning about specific problems, and providing
  new advantageous points of view.
\end{abstract}

\section{Introduction}
\label{sec:intro}

Membrane computing and P systems are a paradigm of massively parallel
computing introduced more than two decades ago by
Gh. P\u{a}un~\cite{Paun2000}, and inspired by the structure and the
functioning of the biological cell.  Following the example of the cell,
a membrane (P) system is a hierarchical membrane structure with
compartments containing multisets of objects, representing in an
abstract sense the biochemical species.  Multiset rewriting rules are
attached to every membrane to represent the reactions.  Over the last
two decades, a considerable number of variants of P systems have been
introduced, inspired by various aspects of cellular life, or capturing
specific computing properties.  For comprehensive overviews we refer
the reader to \cite{imcs,handMC}.

Even though P systems are directly inspired by the biological cell,
their use for actual cellular modelling has encountered relatively
little success.  On the other hand, Boolean networks have been quite
successful recently, despite their relative dissimilarity to
biological structures—a Boolean network is a set of Boolean variables
equipped with Boolean update functions, describing how to compute the
new value of the variables from their current values.  We refer the
reader to~\cite{AlhazovFI2022} for a more detailed impression.

One application of interest of Boolean networks is controllability—the
problem of deciding whether externally modifying some parameters of
a system can make it reach a particular state, and finding the
necessary
modifications~\cite{BianeD19,FumiaM2013,PaulSPM2018,SuPP2019,VonDerHeydeBHSKB2014}.
A variant of this problem which has attracted particular attention is
sequential controllability: instead of looking for a particular
combination of control inputs, find a \emph{sequence} of control
inputs to guide the system to a given
state~\cite{Mandon19,MandonHP17,MandonSH0P19,MandonSPPHP19,Pardo22,PardoID21}.
Sequential controllability is promising because it may allow reducing
the total number of control actions, or may even drive the Boolean
network along trajectories which would otherwise be inaccessible.
On the other hand, sequential controllability is
\PSPACE-hard~\cite{PardoID21}, making it a difficult problem
to tackle.

The goal of this paper is to show how to combine the modelling power
of Boolean networks with the richness of P systems to reason about and
prove some properties of sequential controllability of Boolean
networks.  We construct a P system variant to satisfy the following
two properties simultaneously:
\begin{enumerate}\setlength\itemsep{1mm}
\item represent sequential controllability of Boolean control networks
  via simple syntax transformations,
\item have \PSPACE-complete reachability.
\end{enumerate}
This formalization of sequential controllability allows us to complete
the complexity result from~\cite{PardoID21} by proving that this
problem is \PSPACE-complete.  We would like to use this construction
to promote P system variants as a general tool for building ad hoc
formalisms specifically tailored for tackling particular problems.

This paper is structured as follows.  Section~\ref{sec:prelim} briefly
recalls all the necessary preliminaries: linear bounded automata,
P systems, Boolean networks, sequential controllability.
Section~\ref{sec:boolp} introduces the specific P system variant for
tackling sequential controllability: Boolean P systems.
Section~\ref{sec:boolp-bn} shows how Boolean P~systems can directly
simulate Boolean networks.  Section~\ref{sec:composition} introduces
composition of Boolean P systems in the spirit of automata theory, and
Section~\ref{sec:boolp-seq} shows how composite Boolean P systems can
capture a Boolean network together with the master dynamical system
emitting the control inputs.  In Section~\ref{sec:boolp-reach}, we
show that the reachability problem for Boolean P systems is
\PSPACE-complete, and we use this result in
Section~\ref{sec:seq-pspace} to show that sequential controllability
of Boolean networks is \PSPACE-complete as well.  Finally, in
Section~\ref{sec:conclusion} we extensively discuss the obtained
technical results concerning sequential controllability, the features
of Boolean P systems, and the general methodology of designing ad hoc
formalisms custom-tailored to specific problems.

\section{Preliminaries}
\label{sec:prelim}
In this section, we briefly recall the necessary preliminaries, in
particular deterministic bounded automata, P systems, Boolean
networks, Boolean Control Networks (BCN), and sequential
controllability of BCN.

Given two sets $A$ and $B$, we denote by $B^A$ the set of all
functions $f : A \to B$.  We denote by $2^A$ the set of all subsets of
$A$ (the power set of $A$) and by $|A|$ the cardinal of the set $A$.
An indicator function of a subset $C \subseteq A$ is the function
$i_C : A \to \{0,1\}$ with the property that
$C = \{a \mid i_C(a) = 1\}$.  In this paper, we will often use the
same symbol to refer to a subset and to its indicator function.

\subsection{Deterministic Linear Bounded Automata (LBA)}
A deterministic linear bounded automaton (deterministic LBA or simply
LBA) $\mathcal{M}$ is a construct
\[
  \mathcal{M} = (Q,V,T_1,T_2,\delta, q_0, q_1, Z_l, B, Z_r),
\]
where:
\begin{itemize}
\item $Q$ is a finite set of states,
\item $V$ is the finite tape alphabet,
\item $T_1 \subseteq V\setminus \lbrace Z_l, B,Z_r\rbrace$ is the input alphabet,
\item $T_2 \subseteq V\setminus \lbrace Z_l, B, Z_r \rbrace$ is the output alphabet,
\item $\delta : Q \times V \rightarrow Q \times V \times \lbrace L,R,S
  \rbrace$ is the transition function,
\item $q_0$ is the initial state,
\item $q_1$ is the final state,
\item $Z_l \in V$ is the left boundary marker,
\item $B \in V$ is the blank symbol,
\item $Z_r \in V$ is the right boundary marker,
\end{itemize}
We restrict the transition function such that the automaton can never
write over the boundary markers or exceed them, more precisely:
\[
  \begin{array}{ll}
    \forall q \in Q : & \delta(q,Z_l) \in Q \times \{Z_l\} \times \{R,S\},\text{ and} \\[1mm]
    \forall q \in Q : & \delta(q,Z_r) \in Q \times \{Z_r\} \times \{L,S\}.
  \end{array}
\]

A configuration of the automaton will be written as
$Z_l u \, q \underline{a}\, v Z_r$, where
$a \in V \setminus \{Z_l, Z_r\}$,
$u,v \in (V \setminus \{Z_l, Z_r\})^*$.  The state $q$ is written to
the left of the underlined tape symbol $a$ on which the head of the
automaton currently stands.

Suppose the LBA is in state $q$ and reads the symbol $a$ on the tape.
If $\delta(q,a) = (p,b,D)$, one of the following transitions occurs,
depending on the value of $D \in \{L, R, S\}$:
\[
  \begin{array}{lll ll}
    Z_l u c \, q \underline{a}\, v Z_r &\Rightarrow& Z_l u \, p \underline{c}\, b v Z_r, &&\mbox{ if } D = L,\mbox{ where } c \in V,\\[1mm]
    Z_l u \, q \underline{a}\, c v Z_r &\Rightarrow& Z_l u b \, p \underline{c} \, v Z_r, &\hspace{3mm}&\mbox{ if } D = R, \mbox{ where } c \in V, \\[1mm]
    Z_l u \, q \underline{a}\, v Z_r &\Rightarrow& Z_l u \,p \underline{b}\, v Z_r, &&\mbox{ if } D = S. \\[1mm]
  \end{array}
\]

Due to the restriction of the transition function, the accessible part
of the tape is limited to the input plus the two delimiters $Z_l$ and
$Z_r$. Another model of LBA consists in restricting the size of the
accessible part of the tape to a linear function of the input, which
is the origin of the name \emph{linear} bounded automaton. The two
models have the same computational power~\cite{gareyjohnson}.

An LBA accepts the input $x \in V^*$ if starting with the
configuration $Z_l q_0 x Z_r$ it reaches a configuration of the form
$Z_l q_1 \{B\}^* Z_r$.  Given an LBA $\mathcal{M}$ and an input $x$,
the $\LBAACCEPTANCE$ problem consists in deciding whether
$\mathcal{M}$ accepts~$x$.  This problem is \PSPACE-complete
\cite{gareyjohnson}.

\subsection{P Systems}
In this subsection, we give a general overview of P systems.  For more
details, we refer the reader to~\cite{imcs,handMC}.  A P system is
a construct
\[
  \Pi = (O, T, \mu, w_1,\ldots,w_n, R_1,\ldots R_n, h_i, h_o),
\]
where $O$ is the alphabet of objects, $T\subseteq O$ is the alphabet
of terminal objects, $\mu$ is the membrane structure injectively
labelled by the numbers from $\{1,\ldots,n\}$ and usually given by
a sequence of correctly nested brackets, $w_i$ are the multisets
giving the initial contents of each membrane $i$ ($1\leq i\leq n$),
$R_i$ is the finite set of rules associated with membrane $i$
($1\leq i\leq n$), and $h_i$ and $h_o$ are the labels of the input and
the output membranes, respectively ($1\leq h_i\leq n$,
$1\leq h_o\leq n$).

Quite often, the rules associated with membranes are multiset
rewriting rules (or special cases of such rules).  Multiset rewriting
rules have the form $u\to v$, with
$u\in O^\circ\setminus\{\textbf{0}\}$ and $v\in O^\circ$, where
$O^\circ$ is the set of multisets over $O$, and $\mathbf{0}(a) = 0$,
for all $a \in O$.  If $|u| = 1$, the rule $u\to v$ is called
non-cooperative; otherwise it is called cooperative.  In communication
P systems, rules are additionally allowed to send symbols to the
neighbouring membranes.  In this case, for rules in $R_i$,
$v\in (O\times \mathit{Tar}_i)^\circ$, where $\mathit{Tar}_i$ contains
the symbols $\mathit{out}$ (corresponding to sending the symbol to the
parent membrane), $\mathit{here}$ (indicating that the symbol should
be kept in membrane $i$), and $\mathit{in}_h$ (indicating that the
symbol should be sent into the child membrane $h$ of membrane $i$).
When writing out the multisets over $O \times \mathit{Tar}_i$, the
indication $\mathit{here}$ is often omitted.

In P systems, rules are often applied in a maximally parallel way: in
one derivation step, only a non-extendable multiset of rules can be
applied.  The rules are not allowed to consume the same instance of
a symbol twice, which creates competition for objects and may lead to
non-deterministic choice between the maximal collections of rules
applicable in one step.

A computation of a P system is traditionally considered to be
a sequence of configurations it can successively visit, stopping at
the halting configuration.  A halting configuration is a configuration
in which no rule can be applied any more, in any membrane.  The result
of a computation of a P system $\Pi$ as defined above is the contents
of the output membrane $h_o$ projected over the terminal alphabet $T$.
\begin{example}
  Figure~\ref{fig:example-p} shows the graphical representation of the
  P system formally given by
  \[
    \begin{array}{lcl}
      \Pi & = & (\{a, b, c, d\}, \{a,d\}, [_1[_2]_2]_1, R_1, R_2, 1, 1), \\
      R_2 & = & \{a \to aa, b \to b\, (c,\mathit{out})\},                \\
      R_1 & = & \emptyset.
    \end{array}
  \]

  \begin{figure}[h]
    \centering
    \begin{tikzpicture}
      \tikzstyle membrane=[draw,rounded corners=1mm,inner sep=2mm]
      \node[membrane,align=left] (mem2) {$a \to aa$\\
        $b \to b \, (c,\mathit{out})$\\[2mm]
        \hspace{7mm}$ab$};
      \node[right=0mm of mem2.south east] (lab2) {\scriptsize 2};
      \node[right=3mm of mem2] (skin content) {$d$};
      \node[membrane,fit={(mem2) (lab2) (skin content)}] (skin) {};
      \node[right=0mm of skin.south east] {\scriptsize 1};
    \end{tikzpicture}
    \caption{An example of a simple P system.}
    \label{fig:example-p}
  \end{figure}
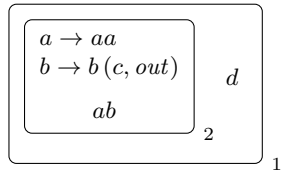

  In the maximally parallel mode, the inner membrane 2 of $\Pi$ will
  apply as many instances of the rules as possible, thereby doubling
  the number of $a$, and ejecting a copy of $c$ into the surrounding
  (skin) membrane at each step.  The symbol $d$ in the skin membrane
  is not used.  Therefore, after $k$ steps of evolution, membrane~2
  will contain the multiset $a^{2^k} b$ and membrane 1 the multiset
  $c^k d$.  Since all rules are always applicable in $\Pi$, this
  P system never halts. \qed
\end{example}

\subsection{Boolean Networks}
A Boolean variable is a variable which may only have values in the
Boolean domain $\{0,1\}$.  Let $X$ be a finite set of Boolean
variables.  A state of these variables is any function
$s : X \to \{0, 1\}$, $s \in \{0, 1\}^X = S_X$ assigning a Boolean
value to every single variable.  An update function is a Boolean
function computing a Boolean value from a state:
$f : S_X \to \{0,1\}$.  A Boolean network over $X$ is a function
$F : S_X \to S_X$, in which the update function for a variable
$x \in X$ is computed as a projection of $F$: $f_x(s) = F(s)_x$.

A Boolean network $F$ computes trajectories on states by updating its
variables according to a (Boolean) mode $M \subseteq 2^X$, defining
the variables which should be updated together in a step.
Typical examples of modes are the synchronous mode
$\mathit{syn} = \{X\}$ and the asynchronous mode
$\mathit{asyn} = \{\{x\} \mid x \in X\}$.  A trajectory $\tau$ of
a Boolean network under a given mode $M$ is any finite sequence of
states $\tau = (s_i)_{0 \leq i \leq n}$ such that $F$ can derive
$s_{i+1}$ from $s_i$ under the mode $M$.

\begin{remark}
  The definition of modes and evolution are quite different in
  P systems and Boolean networks.  The asynchronous mode in Boolean
  networks only allows updating one variable at a time, while the
  asynchronous mode in P systems generally allows any combinations of
  updates.  Furthermore, no halting conditions are generally
  considered in Boolean networks, and the asymptotic behavior is often
  looked at as the important part of the dynamics. \qed
\end{remark}

\begin{example}\label{ex:bool}
  Consider the set of variables $X = \{x, y\}$ with the corresponding
  update functions $f_x(x, y) = \bar x \wedge y$ and
  $f_y(x, y) = x \wedge \bar y$.  Figure~\ref{fig:example-bool} shows
  the possible state transitions of this network under the synchronous
  and the asynchronous modes.  The states are represent as pairs of
  binary digits, e.g. $01$ stands for the state in which $x = 0$ and
  $y = 1$.

  \begin{figure}[h]
    \centering
    \begin{tikzpicture}[node distance=5mm,baseline]
      \node (00) {$00$};
      \node[right=of 00] (01) {$01$};
      \node[below=of 01] (10) {$10$};
      \node[below=of 00] (11) {$11$};

      \draw[->] (00) to[out=150,in=-150,looseness=4] (00);
      \draw[->] (11) to (00);
      \draw[<->] (01) to (10);
    \end{tikzpicture}
    \hspace{15mm}
    \begin{tikzpicture} [node distance=5mm,baseline]
      \node (00) {$00$};
      \node[right=of 00] (01) {$01$};
      \node[below=of 01] (10) {$10$};
      \node[below=of 00] (11) {$11$};

      \draw[->] (00) to[out=150,in=-150,looseness=4] (00);
      \draw[->] (01) to (00);
      \draw[->] (10) to (00);

      \draw[<->] (01) to (11);
      \draw[<->] (10) to (11);
    \end{tikzpicture}
    \caption{The synchronous (left) and the asynchronous (right)
      dynamics of the Boolean network in Example~\ref{ex:bool}.}
    \label{fig:example-bool}
  \end{figure}
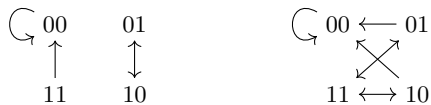

  We notice that, under the synchronous mode, this network exhibits
  three kinds of behaviors.  If initialized to $00$, it will stay in
  this state forever---$00$ is a stable state.  If initialized to
  $11$, the network will directly converge to $00$.  Finally, if it is
  initialized to any one of the states $01$ or $10$, it will oscillate
  between them.  The synchronous mode yields deterministic behavior.

  The state transitions are quite different under the asynchronous
  mode, under which only one variable may be updated at a time.
  While state $00$ remains stable, states $01$ and $10$ can now
  oscillate to $11$, but not directly between them.  Moreover, these
  states can also converge to $00$, but $11$ cannot anymore. \qed
\end{example}

\subsection{Boolean Control Networks (BCN)}
Boolean networks are often used to represent biological networks in
the presence of external perturbations: environmental hazards, drug
treatments, etc. (e.g.,~\cite{Barabasi2011,BianeD19,PardoID21}). To
represent network reprogramming, an extension of Boolean networks can
be considered: Boolean control networks (BCN)~\cite{BianeD19}.
Informally, a BCN is a parameterized Boolean network template;
assigning a Boolean value to every single one of its parameters yields
a Boolean network.

Formally, a Boolean control network is a function
$F_U : S_U \to (S_X \to S_X)$, where the elements of $U$,
$U \cap X = \emptyset$, are called the control inputs.  To every
valuation of control inputs, $F_U$ associates a Boolean network.
A control $\mu$ of $F_U$ is any Boolean assignment to the control
inputs: $\mu : U \to \{0, 1\}$.

\smallskip

While this definition of BCNs is very general, in practice one
restricts the impact the control inputs may have on the BCN to some
biologically relevant classes.  One particularly useful class are
freeze perturbations, in which a variable in $X$ is temporarily frozen
to 0 or to 1, independently of its normal update function.
These actions mean to model gene knock-outs and knock-ins.

When Boolean update functions are written as propositional formulae,
freeze control inputs can be written directly in the formulae.
Consider for example a Boolean network $F$ over $X = \{x_1, x_2\}$
with the update functions $f_1 = x_1 \wedge x_2$ and $f_2 = x_2$.
To allow for freezing of $x_1$, we introduce the control variables
$U = \{u_1^0, u_1^1\}$ into the Boolean formula of $f_1$ in the
following way: $f'_1 = (x_1 \wedge x_2) \wedge u_1^0 \vee \overline{u_1^1}$.
Setting $u_1^0$ to 0 and $u_1^1$ to 1 freezes $x_1$ to 0,
independently of the values of $x_1$ and $x_2$.  Symmetrically,
setting $u_1^1$ to 0 and $u_1^0$ to 1 freezes $x_1$ to 1.
Setting both $u_1^0$ and $u_1^1$ to 0 is generally disallowed.

In this paper, we will use two notations to indicate which control
inputs correspond to which controlled variable.  In the simplest
examples in which the variables have no indices, e.g. $x$ or $y$, we
will directly specify the name of the variable in the subscript of the
corresponding control inputs, like so: $u_x^0$, $u_x^1$, $u_y^0$, or
$u_y^1$.  In more general cases, we will refer to the variables by
indexed names $x_i$, and we will only specify the respective index as
the subscript of the corresponding control inputs: $u_i^0$ and
$u_i^1$.

\subsection{Sequential Controllability of BCN}
\label{sec:seq-bcn}

In many situations, perturbations of biological networks do not happen
once, but rather accumulate or evolve over
time~\cite{Fearon1990,Lee2012,PardoID21}.  In the language of Boolean
control networks, this accumulation can be represented by
sequences of controls $(\mu_1, \dots, \mu_n)$.  More precisely, take
a BCN $F_U$ with the variables $X$ and the control inputs $U$, as well
as a sequence of controls $\mu_{[n]} = (\mu_1, \dots, \mu_n)$,
$\mu_i : U \to \{0, 1\} \in S_U$.  This gives rise to a sequence of
Boolean networks $\left(F_U(\mu_1), \dots, F_U(\mu_n)\right)$.
Fix a mode $M$ and consider a sequence of trajectories
$(\tau_1, \dots, \tau_n)$ of these Boolean networks.  Such a sequence
is an evolution of $F_U$ under the sequence of controls $\mu_{[n]}$ if
the last state of every $\tau_i$ is the first state of $\tau_{i+1}$.
In this case we can speak of the trajectory of the BCN $F_U$ under the
control sequence $\mu_{[n]}$ as the concatenation of the individual
trajectories $\tau_i$, in which the last state of every single
$\tau_i$ is glued together with the first state of~$\tau_{i+1}$.

Given the 3-tuple $(F_U, S_\alpha, S_\omega)$, where $F_U$ is a BCN,
$S_\alpha$ is a set of starting states, and $S_\omega$ is a set of
target states, the sequence inference problem is the problem of
inferring a control sequence driving $F_U$ from each state in
$S_\alpha$ to any state in $S_\omega$.  This problem was called the
CoFaSe problem in~\cite{PardoID21} and was extensively studied.
In particular, is was shown that CoFaSe is \PSPACE-hard.

\begin{example}\label{ex:bcn}
  Consider again the Boolean network from Example~\ref{ex:bool}, with
  $X = \{x, y\}$ and the update functions $f_x = \bar x \wedge y$ and
  $f_y = x \wedge \bar y$.  As mentioned before, a convenient way to
  express freezing controls is by explicitly including the control
  inputs into the update functions in the following way:
  \[
    \renewcommand{\arraystretch}{1.3}
    \begin{array}{lcl}
      f'_x & = & (\bar x \wedge y) \wedge u_x^0 \vee \overline{u_x^1}, \\
      f'_y & = & (x \wedge \bar y) \wedge u_y^0 \vee \overline{u_y^1}. \\
    \end{array}
  \]
  Notice how setting $u_x^0$ to 0 essentially sets $f_x' = 0$, and
  setting $u_x^1$ to 0 essentially sets $f_x' = 1$, independently of
  the actual value of $x$ or $y$.

  \smallskip

  Consider now the following 3 controls:
  \[
    \renewcommand{\arraystretch}{1.3}
    \begin{array}[lcl]{lcl}
      \mu_1 & = & \{u_x^0 \leftarrow 1, u_x ^1 \leftarrow 1, u_y^0 \leftarrow 1, u_y^1 \leftarrow 1\},             \\
      \mu_2 & = & \{\underline{u_x^0 \leftarrow 0}, u_x ^1 \leftarrow 1, u_y^0 \leftarrow 1, u_y^1 \leftarrow 1\}, \\
      \mu_3 & = & \{u_x^0 \leftarrow 1, u_x ^1 \leftarrow 1, u_y^0 \leftarrow 1, \underline{u_y^1 \leftarrow 0}\}. \\
    \end{array}
  \]
  Informally $\mu_1$ does not freeze any variables, $\mu_2$ freezes
  $x$ to 0, and $\mu_3$ freezes $y$ to~1.  Consider now the BCN $F_U$
  with the variables $X = \{x,y\}$ and the controlled update functions
  $f_x'$ and $f_y'$.  Fix the synchronous update mode.  A trajectory
  of this BCN under the control $\mu_1$---i.e. a trajectory of
  $F_U(\mu_1)$---is $\tau_1 : 01 \to 10 \to 01$.  A trajectory of
  $F_U(\mu_2)$ is $\tau_2 : 01 \to 00 \to 00$; remark that $00$ is
  still a stable state of $F_U(\mu_2)$.  A trajectory of $F_U(\mu_3)$
  is $\tau_3 : 00 \to 01 \to 11$.  We can now glue together the
  trajectories $\tau_1$, $\tau_2$, and $\tau_3$ by identifying their
  respective ending and starting states, and we obtain the following
  trajectory of the BCN $F_U$ under the control sequence
  $\mu_{[3]} = (\mu_1, \mu_2, \mu_3)$:
  \[
    \tau : 01 \to 10 \to 01 \to 00 \to 00 \to 01 \to 11.
  \]
  It follows from this construction that $\mu_{[3]}$ is a solution for
  the CoFaSe problem $(F_U, \{01\}, \{11\})$.  Remark that $11$ is not
  reachable from $01$ under the synchronous mode in the uncontrolled case, as
  Figure~\ref{fig:example-bool} illustrates. \qed
\end{example}

\begin{remark}
  We follow the approach from~\cite{PardoID21} which decorrelates the
  length of the control sequence from the length of the trajectories
  it yields.  Thus, $\mu_{[3]}$ can yield trajectories of different
  lengths greater or equal to 3.  From the modeling standpoint, this
  represents the fact that the time scale on which control inputs are
  emitted is not necessarily the same as the time scale of the
  controlled system.
\end{remark}

\section{Boolean P Systems}
\label{sec:boolp}
In this section we introduce a new variant of P systems---Boolean
P systems---tailored specifically to capture sequential
controllability of Boolean networks with as little descriptional
overhead as possible.  We further tackle the differences between
evolution modes in Boolean networks and P systems by
introducing quasimodes.

Rather than trying to be faithful to the original model of P systems
as recalled in Section~\ref{sec:prelim}, we here invoke the intrinsic
flexibility of the domain to design a variant fitting to our specific
use case.

\subsection{Formalism}
Boolean P systems are set rewriting systems.  A Boolean state
$s : X \to \{0, 1\}$ is represented as the subset of $X$ obtained by
considering $s$ as an indicator function:
$\{x \in X \mid s(x) = 1 \}$.  By abuse of notation, we will sometimes
use the symbol $s$ to refer both to the Boolean state and to the
corresponding subset of $X$.

A Boolean P system is the following construct:
\[
  \Pi = (V, R),
\]
where $V$ is the alphabet of symbols, and $R$ is a set of rewriting
rules with propositional guards.  A rule $r \in R$ is of the form
\[
  r : A \to B \mid \varphi,
\]
where $A, B \subseteq X$ and $\varphi$ is the guard---a propositional
formula with variables from $V$.  The rule $r$ is applicable to a set
$W \subseteq V$ if $A \subseteq W$ and $W \in \varphi$, where by abuse
of notation we use the same symbol $\varphi$ to indicate the set of
subsets of $V$ which satisfy $\varphi$.  Formally, for
$W \subseteq V$, we denote by $\varphi(W)$ the truth value of the
formula obtained by replacing all variables appearing in $W$ by
1 in $\varphi$, and by 0 all variables from $V \setminus W$.  Then the
set of subsets satisfying $\varphi$ is
$\varphi = \{W \subseteq V \mid \varphi(W) \equiv 1\}$.

Applying the rule $r : A \to B \mid \varphi$ to a set $W$ results in
the set $(W \setminus A) \cup B$.  Applying a finite set of separately
applicable rules $\{r_i : A_i \to B_i \mid \varphi_i\}$ to $W$ results
in the new set
\[
  \left(W \setminus \bigcup_i A_i \right) \cup \bigcup_i B_i.
\]
Note how this definition excludes competition between the rules, as
only individual applicability is checked.  Further note that applying
a rule multiple times to the same configuration has exactly the same
effect as applying it once.

In P systems, the set of multisets of rules of $\Pi$ applicable to
a given configuration $W$ is usually denoted by
$\mathit{Appl}(\Pi, W)$~\cite{FreundV2007}.  Since in Boolean
P systems multiple applications of rules need not be considered, we
will only look at the set of \emph{sets} of rules applicable to
a given configuration $W$ of a Boolean P system $\Pi = (V, R)$, and
use the same notation $\mathit{Appl}(\Pi, W)$.  A mode $M$ of $\Pi$
will then be a function assigning to any configuration $W$ of $\Pi$
a set of sets of rules applicable to $W$:
$M(W) \subseteq \mathit{Appl}(\Pi, W)$.
If $|M(W)| \leq 1$ for any $W \subseteq V$, the mode $M$ is called
deterministic\footnote{More precisely, this is the definition of
  strong determinism, see~\cite{AlhazovFM12}.}.  Otherwise it is
called non-deterministic.

An evolution of $\Pi$ under the mode $M$ is a sequence of states
$(W_i)_{0 \leq i \leq k}$ with the property that $W_{i+1}$ is obtained
from $W_i$ by applying one of the sets of rules $R' \in M(W_i)$
prescribed by the mode $M$ in state $W_i$.  This is usually written as
$W_i \overset{R'}{\longrightarrow} W_{i+1}$.  If no rules are
applicable in state $W_k$, it is called the halting state, and
$(W_i)_{0 \leq i \leq k}$ is called a halting evolution.

\begin{example}
  Take $V = \{a, b\}$ and consider the following rules
  $r_1 : \{a, b\} \to \{a\} \mid \mathbf{1}$ and
  $r_2 : \{a\} \to \emptyset \mid \bar b$, where $\mathbf{1}$ is the
  Boolean tautology.  Construct the Boolean P system
  $\Pi = (V, \{r_1, r_2\})$.  Informally, $r_1$ removes $b$ from
  a configuration which contains $a$ and $b$, and $r_2$ removes $a$
  from the configuration which does not already contain $b$.
  A possible trajectory of $\Pi$ under the maximally parallel
  mode---which applies non-extendable applicable sets of rules---is
  $\{a, b\} \to \{a\} \to \emptyset$.  Note that only $r_1$ is
  applicable in the first step, since $r_2$ requires the configuration
  to not contain $b$. \qed
\end{example}

\begin{remark}\label{rem:boolp-rs}
  Boolean P systems as defined here are very close to other set
  rewriting formalisms, and in particular to reaction
  systems~\cite{EhrenfeuchtR2007}.  A reaction system $\mathcal A$
  over a set of species $S$ is a set of reactions (rules) of the form
  $a : (R_a, I_a, P_a)$, in which $R_a \subseteq S$ is called the set
  of reactants, $I_a \subseteq S$ the set of inhibitors, and
  $P_a \subseteq S$ the set of products.  For $a$ to be applicable to
  a set $W$, it must hold that $R_a \subseteq W$ and
  $I_a \cap W = \emptyset$.  Applying such a reaction to $W$ yields
  $P_a$, i.e. the species which are not explicitly sustained by the
  reactions disappear.

  We claim that despite their apparent similarity and tight
  relationship with Boolean functions, reaction systems are not so
  good a fit for reasoning about Boolean networks as Boolean
  P systems.  In particular:
  \begin{enumerate}
    \item Reaction systems lack modes and therefore non-determinism,
          which may appear in Boolean networks under the asynchronous
          Boolean mode.
    \item The rule applicability condition is more powerful in Boolean
          P systems, and closer to Boolean functions than in
          reaction systems.
    \item Symbols in reaction systems disappear unless sustained by
          a rule, which represents the degradation of species in
          biochemistry, but which makes reaction systems harder to use to
          directly reason about Boolean networks.
  \end{enumerate}
  We recall that our main goal behind introducing Boolean P systems is
  reasoning about Boolean networks in a more expressive framework.
  This means that zero-overhead representation of concepts from
  Boolean networks is paramount. \qed
\end{remark}

\begin{remark}\label{rem:rs-control}
  Reaction systems~\cite{EhrenfeuchtR2007} are intrinsically
  interesting for discussing controllability, because they are defined
  as open systems from the start, via the explicit introduction of
  context.  Note however that contexts only allow adding symbols to
  the configuration, not removing them.  We refer to~\cite{IvanovP20}
  for an in-depth discussion of controllability of reaction
  systems. \qed
\end{remark}

\subsection{Quasimodes}
\label{sec:quasimodes}

An update function in a Boolean network can always be computed, but
a rule in a Boolean P system need not always be applicable.  This is
the reason behind the difference in the way modes are defined in the
two formalisms: in Boolean networks a mode is essentially a set of
subsets of update functions, while in Boolean P~systems a mode is
a function incorporating applicability checks.  This means in
particular that Boolean network modes are not directly transposable to
Boolean P systems.

To better bridge the two different notions of modes, we introduce
quasimodes.  A \emph{quasimode} $\tilde M$ of a P system
$\Pi = (V, R)$ is any set of sets of rules: $\tilde M \subseteq 2^R$.
The mode $M$ corresponding to the quasimode $\tilde M$ is derived in
the following way:
\[
  M(W) = \tilde M \cap \mathit{Appl}(\Pi, W).
\]
Given a configuration $W$ of $\Pi$, $M$ picks only those sets of rules
from $\tilde M$ which are also applicable to $W$.  Thus, instead of
explicitly giving the rules to be applied to a given configuration of
a P system $W$, a quasimode advises the rules to be applied.

In the rest of the paper, we will say ``evolution of $\Pi$ under the
quasimode~$\tilde M$'' to mean ``evolution of $\Pi$ under the mode
derived from the quasimode $\tilde M$''.

\section{Boolean P Systems Simulate Boolean Networks}
\label{sec:boolp-bn}

Consider a Boolean network $F$ over the set of variables $X$, and take
a variable $x \in X$ with its corresponding update function $f_x$.
The update function $f_x$ can be simulated by two Boolean P systems
rules: the rule corresponding to setting $x$ to 1, i.e. introducing
$x$ into the configuration, and the rule corresponding to setting $x$
to 0, i.e. erasing $x$ from the configuration:
\[
  R_x = \{\;\; \emptyset \to \{x\} \mid f_x, \;\; \{x\} \to \emptyset
  \mid \overline{f_x\strut} \;\;\}.
\]
Consider now the following Boolean P system:
\[
  \Pi(F) = \left(X, \bigcup_{x \in X} R_x \right).
\]
We claim that $\Pi(F)$ faithfully simulates $F$.

\begin{theorem}\label{thm:bp-bn}
  Take a Boolean network $F$ and a Boolean mode $M$.  Then the Boolean
  P system $\Pi(F)$ constructed as above and working under the
  quasimode
  $\tilde M = \left\{\bigcup_{x \in m} R_x \mid m \in M\right\}$
  faithfully simulates $F$: for any evolution of $F$ under $M$ there
  exists an equivalent evolution of $\Pi(F)$ under $\tilde M$, and
  conversely, for any evolution of $\Pi(F)$ under $\tilde M$ there
  exists an equivalent evolution of $F$ under~$M$.
\end{theorem}

\begin{proof}
  Take two arbitrary states $s$ and $s'$ of $F$ such that $s'$ is
  reachable from $s$ by the update prescribed by an element $m \in M$.
  Consider now the subsets of variables $W, W' \subseteq X$ defined by
  $s$ and $s'$ taken as the respective indicator functions.
  It follows from the construction of $\tilde M$ that it contains an
  element $\tilde m$ including the update rules for all the variables
  of $m$: $\tilde m = \bigcup_{x \in m} R_x$.  Therefore, $\Pi(F)$ can
  derive $W'$ from $W$ under the quasimode $\tilde M$.

  Conversely, consider two subsets of variables $W, W' \subseteq X$
  such that $\Pi(F)$ can derive $W'$ from $W$ under the update
  prescribed by an element $\tilde m \in \tilde M$.  By construction
  of $\tilde M$, there exists a subset $m \subseteq X$ such that
  $\tilde m = \bigcup_{x \in m} R_x$.  Take now the indicator
  functions $s, s' : X \to \{0, 1\}$ describing $W$ and $W'$
  respectively.  Then $F$ can derive $s'$ from $s$ by updating the
  variables in $m$.

  We conclude that the transitions of $\Pi(F)$ exactly correspond to
  the transitions of $F$, which proves the statement of the
  theorem. \qed
\end{proof}

\begin{example}
  Consider the Boolean network $F_U$ from Example~\ref{ex:bool}:
  \[
    \begin{array}{lcl}
      f_x & = & \bar x \wedge y, \\
      f_y & = & x \wedge \bar y. \\
    \end{array}
  \]
  This Boolean network can be translated to the Boolean P system
  $\Pi = (V, R)$ with $V = \{x, y\}$ and the following rules:
  \[
    \begin{array}{lcl}
      R &=& R_x \cup R_y, \\[1mm]
      R_x &=& \{\; \emptyset \to \{x\} \mid \bar x \wedge y, \;
              \{x\} \to \emptyset \mid \overline{\bar x \wedge y\strut} \;\}, \\[1.2mm]
      R_y &=& \{\; \emptyset \to \{y\} \mid x \wedge \bar y, \;
                 \{y\} \to \emptyset \mid \overline{x \wedge \bar y\strut}\; \}.
    \end{array}
  \]
  The first rule in $R_x$ ensures that $x$ is introduced whenever the
  current state satisfies $\bar x \wedge y = f_x$, and the second rule
  ensures that $x$ is removed whenever the current state does not
  satisfy $\bar x \wedge y$.  Similarly, the rules in $R_y$ introduce
  or remove $y$ depending on whether the current state satisfies
  $f_y$.

  To simulate $F_U$ under the Boolean synchronous mode, $\Pi$ should
  run under the quasimode $\tilde M_\textit{syn} = \{R\}$, i.e. the
  quasimode allowing all rules in $R$ to be applied at all times.
  To simulate $F_U$ under the Boolean asynchronous mode, $\Pi$ should
  run under the quasimode $\tilde M_\textit{asyn} = \{R_x, R_y\}$,
  i.e. the quasimode allowing the application of either both rules in
  $R_x$, or both rules in $R_y$, but not all 4 rules at a time.  \qed
\end{example}

\begin{remark}
  Incidentally, Boolean P systems also capture reaction systems (see
  also Remarks~\ref{rem:boolp-rs} and~\ref{rem:rs-control}).  Indeed,
  consider a reaction $a = (R_a, I_a, P_a)$ with the reactants $R_a$,
  inhibitors $I_a$, and products $P_a$.  It can be directly simulated
  by the Boolean P system rule $\emptyset \to P_a \mid \varphi_a$,
  where
  $\varphi_a = \bigwedge_{x \in R_a} x \wedge \bigwedge_{y \in I_a}
  \bar y$.  The degradation of the species in reaction systems can be
  simulated by adding a rule $\{x\} \to \emptyset \mid \mathbf{1}$ for
  every species $x$, where $\mathbf{1}$ is the Boolean tautology. \qed
\end{remark}

\section{Composition of Boolean P Systems}
\label{sec:composition}

In this section, we define the composition of Boolean P systems in the
spirit of automata theory.  Consider two Boolean P systems
$\Pi_1 = (V_1, R_1)$ and $\Pi_2 = (V_2, R_2)$.  We will call the union
of $\Pi_1$ and $\Pi_2$ the Boolean P system
$\Pi_1 \cup \Pi_2 = (V_1 \cup V_2, R_1 \cup R_2)$.  Note that the
alphabets $V_1$ and $V_2$, as well as the rules $R_1$ and $R_2$ are
not necessarily disjoint.

To talk about the evolution of $\Pi_1 \cup \Pi_2$, we first define
a variant of Cartesian product of two sets of sets $A$ and $B$:
$A \dottimes B = \{a \cup b \mid a \in A, b \in B\}$.  We remark now
that
\[
  \forall W\subseteq V_1 \cup V_2 : \mathit{Appl}(\Pi_1 \cup \Pi_2, W)
  = \mathit{Appl}(\Pi_1, W) \dottimes \mathit{Appl}(\Pi_2, W).
\]
Indeed, since the rules of Boolean P systems do not compete for
resources among them, the applicability of any individual rule is
independent of the applicability of the other rules.  Therefore, the
applicability of a set of rules of $\Pi_1$ to a configuration $W$ is
independent of the applicability of a set of rules of $\Pi_2$ to $W$.

For a mode $M_1$ of $\Pi_1$ and a mode $M_2$ of $\Pi_2$, we define
their product as follows:
\[
  (M_1 \times M_2)(W) = M_1(W) \dottimes M_2(W).
\]
The union of Boolean P systems $\Pi_1 \cup \Pi_2$ together with the
product mode $M_1 \times M_2$ implement parallel composition of the
two P systems.  In particular, if the alphabets of $\Pi_1$ and $\Pi_2$
are disjoint, the projection of any evolution of $\Pi_1 \cup \Pi_2$
under the mode $M_1 \times M_2$ on the alphabet $V_1$ will yield
a valid evolution of $\Pi_1$ under $M_1$ (modulo some repeated
states), while the projection on $V_2$ will yield a valid evolution of
$\Pi_2$ under the mode $M_2$ (modulo some repeated states).  Note that
this property may not be true if the two alphabets intersect
$V_1 \cap V_2 \neq \emptyset$.

Quasimodes fit naturally with the composition of modes, as the
following lemma shows.

\begin{lemma}
  If the mode $M_1$ can be derived from the quasimode $\tilde M_1$ and
  $M_2$ from the quasimode $\tilde M_2$, then the product mode
  $M_1 \times M_2$ can be derived from
  $\tilde M_1 \dottimes \tilde M_2$:

  \begin{center}
    \begin{tikzpicture}[node distance=7mm and 6mm]
      \node (M1-x-M2) {$M_1 \times M_2$};
      \node[above=of M1-x-M2] (tM1-x-tM2) {$\tilde M_1 \dottimes \tilde M_2$};
      \node[yshift=-3mm,base left=of tM1-x-tM2] (tM1) {$\tilde M_1$};
      \node[yshift=-3mm,base right=of tM1-x-tM2] (tM2) {$\tilde M_2$};
      \node[yshift=-1mm,below=of tM1] (M1) {$M_1$};
      \node[yshift=-1mm,below=of tM2] (M2) {$M_2$};

      \draw[->,densely dashed] (tM1) -- (M1);
      \draw[->,densely dashed] (tM2) -- (M2);
      \draw[->,densely dashed] (tM1-x-tM2) -- (M1-x-M2);

      \draw[<-] (M1) -- (M1-x-M2);
      \draw[<-] (M2) -- (M1-x-M2);
      \draw[<-] (tM1) -- (tM1-x-tM2);
      \draw[<-] (tM2) -- (tM1-x-tM2);
    \end{tikzpicture}
  \end{center}
  where a dashed arrow \tikz\draw[->,densely dashed] (0,0) -- (5mm,0);
  from a quasimode to a mode indicates that the mode is derived from
  the quasimode, and the arrows \tikz\draw[->] (0,0) -- (5mm,0); are
  the respective projections.
\end{lemma}

\begin{proof}
  Consider the mode $M_{12}$ derived from $\tilde M_1 \dottimes \tilde
  M_2$:
  \[
    M_{12}(W) = \left(\tilde M_1 \dottimes \tilde M_2\right) \cap
    \mathit{Appl}(\Pi, W).
  \]
  Pick an arbitrary element $m_{12} \in M_{12}(W)$ and remark that it
  can be seen as a union $m = m_1 \cup m_2$ where $m_1$ is a subset of
  applicable rules with the property that $m_1 \in \tilde M_1$, and
  $m_2$ is a subset of applicable rules with the property that
  $m_2 \in \tilde M_2$.  Thus
  $m_1 \in \tilde M_1 \cap \mathit{Appl}(\Pi,W)$ and
  $m_2 \in \tilde M_2 \cap \mathit{Appl}(\Pi,W)$, implying that
  \[
    M_{12}(W) \subseteq
    \left(\tilde M_1 \cap \mathit{Appl}(\Pi,W)\right)
    \dottimes
    \left(\tilde M_2 \cap \mathit{Appl}(\Pi,W)\right).
  \]

  Consider on the other hand an arbitrary
  $m_1 \in \tilde M_1 \cap \mathit{Appl}(\Pi,W)$ and an arbitrary
  $m_2 \in \tilde M_2 \cap \mathit{Appl}(\Pi,W)$.  By definition of
  $\dottimes$,
  $m_1 \cup m_2 \in \tilde M_1 \dottimes \tilde M_2$.  Remark that
  every rule in $m_1$ and $m_2$ is individually applicable, meaning
  that they are also applicable together and that
  $m_1 \cup m_2 \in \mathit{Appl}(\Pi, W)$.  Combining this
  observation with the reasoning from the previous paragraph we
  finally derive:
  \[
    M_{12}(W) =
    \left(\tilde M_1 \cap \mathit{Appl}(\Pi,W)\right)
    \dottimes
    \left(\tilde M_2 \cap \mathit{Appl}(\Pi,W)\right)
    = M_1(W) \dottimes M_2(W),
  \]
  which implies that $M_{12} = M_1 \times M_2$ and concludes the
  proof. \qed
\end{proof}

\section{Boolean P Systems for Sequential Controllability}
\label{sec:boolp-seq}

Underlying sequential controllability of Boolean control networks
(Section~\ref{sec:seq-bcn}) is the implicit presence of a master
dynamical system emitting the control inputs to the network and
thereby driving it.  This master system is external with respect to
the controlled BCN.  The framework of Boolean P systems is
sufficiently general to capture both the master system and the
controlled BCN in a single homogeneous formalism.  In this section, we
show how to construct such Boolean P systems for dealing with
questions of controllability.

\smallskip

Any BCN $F_U : S_U \to (S_X \to S_X)$ can be written as a system of
propositional formulae over $X \cup U$.  First, note that a control
$\mu \in S_U$ can be described by the conjuction
$\bigwedge_{u \in \mu} u \wedge \bigwedge_{v \in U \setminus \mu} \bar
v$.  Now fix an $x \in X$ and consider the formula
\begin{equation}
  \bigvee_{\mu \in S_U} \mu \wedge F(\mu)_x,
  \label{eq:controlled-update-function}
\end{equation}
in which $\mu$ enumerates all the conjuctions corresponding to the
controls in $S_U$ and $F(\mu)_x$ is the propositional formula of the
update function which $F(\mu)$ associates to $x$.
Using (\ref{eq:controlled-update-function}), we can translate any BCN
$F_U : S_U \to (S_X \to S_X)$ into the system of Boolean functions
$F' : S_{X \cup U} \to S_X$ and use
the set $R_x$ from Section~\ref{sec:boolp-bn} to further translate the
individual components of $F'$ to pairs of Boolean P system rules.
Denote $\Pi = (X \cup U, R)$ the Boolean P system whose set of rules
is precisely the union of the sets $R_x$ mentioned above, for
$x\in X$. Finally, construct the Boolean P~system $\Pi_U(U, R_U)$ with
the following rules whose guards are always true:
\[
  \renewcommand{\arraystretch}{1.3}
  \begin{array}{lcl}
    R_U   & = & R_U^0 \cup R_U^1,                                                \\[1mm]
    R_U^0 & = & \{\;\,\{u\} \to \emptyset \mid \mathbf{1} \;\, \mid u \in U\,\}, \\
    R_U^1 & = & \{\;\,\emptyset \to \{u\} \mid \mathbf{1} \;\, \mid u \in U\,\}. \\
  \end{array}
\]

Suppose now that the original BCN $F_U$ runs under the mode $M$, and
consider the corresponding quasimode
$\tilde M = \left\{\bigcup_{x\in m}R_x \mid m \in M\right\}$, as well
as the quasimode
\[
  \tilde M_{U} =  \{R_U^0\} \dottimes 2^{R_U^1}.
\]
Every element $m_U \in \tilde M_U$ is a union of $R_U^0$ and a subset
of $R_U^1$, meaning that $m_U$ enables all rules removing the control
inputs, and enables \emph{some} of the rules adding back
control inputs.

We claim that the Boolean P system $\Pi \cup \Pi_U$ running under the
quasimode $\tilde M \dottimes \tilde M_{U}$ faithfully simulates the
BCN $F_U$ running under the mode $M$.  The following theorem
formalizes this claim.

\begin{theorem}\label{thm:bp-seq}
  Consider a BCN $F_U$ running under the mode $M$. Then the Boolean
  P system $\Pi \cup \Pi_U$ constructed as above and running under the
  quasimode $\tilde M \dottimes \tilde M_U$ faithfully simulates
  $F_U$:
  \begin{enumerate}
    \item For any evolution of $F_U$ under $M$ there exists an
          equivalent evolution of $\Pi \cup \Pi_U$ under
          $\tilde M \dottimes \tilde M_U$.
    \item For any evolution of $\Pi \cup \Pi_U$ under
          $\tilde M \dottimes \tilde M_U$ there exists an equivalent
          evolution of $F_U$ under $M$.
  \end{enumerate}
\end{theorem}

\begin{proof}
  \emph{(1)}\hspace{1ex} Consider two states $s, s' \in S_X$ and
  a control $\mu \in S_U$ such that $F_U(\mu)$ reaches $s'$ from $s$
  in one step.  Take $W, W' \subseteq X$ and $W_U \subseteq U$ by
  respectively taking $s$, $s'$, and $\mu$ as indicator functions.
  Then, as in Theorem \ref{thm:bp-bn}, there exists an
  $\tilde m \in \tilde M$ such that $\Pi$ reaches $W' \cup W_U$ from
  $W \cup W_U$ in one step.  This follows directly from the
  construction of the rules in $\Pi$ and from the fact that $W_U$
  contains exactly the symbols corresponding to the control inputs
  activated by~$\mu$.

  Take now $\tilde M \dottimes \tilde M_U$ and remark that its
  elements are of the form $\tilde m \cup \tilde m_U$, where
  $\tilde m_U = \tilde m_U^1 \cup R_U^0$ and
  $\tilde m_U^1 \subseteq R_U^1$.  Under such an element
  $\tilde m \cup \tilde m_U$, $\Pi \cup \Pi_U$ reaches a state
  $W' \cup W_U'$ from $W \cup W_U$ in one step, where $W_U'$ contains
  the symbols from $U$ introduced by the rules selected by
  $\tilde m_U^1$.  Further note that all elements of $W_U$ are always
  erased by the rules $R_U^0$, but may be immediately reintroduced by
  $m_U^1$.

  Suppose now that $F_U(\mu)$ reaches $s'$ from $s$ in multiple steps.
  Then $\Pi$ reaches $W' \cup W_U$ from $W \cup W_U$ in the same
  number of steps, provided that $\tilde m_U^1$ is always chosen such
  that the rules it selects reintroduce exactly the subset $W_U$.
  If $F_U$ reaches $s'$ from $s$ in multiple steps, but the control
  evolves as well, it suffices to choose $\tilde m_U^1$ such that it
  introduces the correct control inputs before each step.  Finally,
  the control $\mu_0$ applied in the first step of a trajectory of
  $F_U$ must be introduced by setting the starting state of
  $\Pi \cup \Pi_U$ to $W \cup W_U^0$, where $W$ corresponds to the
  initial state of the trajectory of $F_U$.

  \medskip

  \noindent
  \emph{(2)}\hspace{1ex} The converse construction is symmetric.
  A state $W \cup W_U$ of $\Pi \cup \Pi_U$ is translated into the
  state $s \in S_X$ and the control $\mu \in S_U$ corresponding to
  $W_U$.  A step of $\Pi \cup \Pi_U$ under $\tilde m \cup \tilde m_U$
  is translated to applying $\mu$ to $F_U$ and updating the variables
  corresponding to the rules activated by $\tilde m$.  In this way,
  for any trajectory of $\Pi \cup \Pi_U$ under the quasimode
  $\tilde M \dottimes \tilde M_U$ there exists a corresponding
  trajectory in the controlled dynamics of $F_U$. \qed
\end{proof}

We now give an extensive example showing how the composite system
$\Pi \cup \Pi_U$ from the proof above is constructed for a concrete
BCN, and detailing how $\Pi \cup \Pi_U$ simulates its sequentially
controlled trajectories.

\begin{example}\label{ex:pi-pi-u}
  Consider the BCN $F_U$ from Example~\ref{ex:bcn} with the following
  update functions modified to include the control inputs:
  \[
    \renewcommand{\arraystretch}{1.3}
    \begin{array}{lcl}
      f'_x & = & (\bar x \wedge y) \wedge u_x^0 \vee \overline{u_x^1}, \\
      f'_y & = & (x \wedge \bar y) \wedge u_y^0 \vee \overline{u_y^1}, \\
    \end{array}
  \]
  and recall that $X = \{x, y\}$ and
  $U = \{u_x^0, u_x^1, u_y^0, u_y^1\}$.  Since the control inputs are
  already explicitly present in the propositional formulae, we can put
  these together directly to obtain $F' : S_{X \cup U} \to S_X$,
  bypassing equation~\ref{eq:controlled-update-function}.

  \paragraph{Construction of $\Pi \cup \Pi_U.$}
  First construct the Boolean P system $\Pi = (X \cup U, R)$ with the
  following rules:
  \[
    \begin{array}{lcl}
      R &=& R_x \cup R_y, \\[1mm]
      R_x &=& \{\; \emptyset \to \{x\} \mid f'_x, \;
              \{x\} \to \emptyset \mid \overline{f'_x} \;\}, \\[1.2mm]
      R_y &=& \{\; \emptyset \to \{y\} \mid f'_y, \;
                 \{y\} \to \emptyset \mid \overline{f'_y} \;\}.
    \end{array}
  \]
  Now, define $\Pi_U = (U, R_U)$ with the following rules:
  \[
    \renewcommand{\arraystretch}{1.3}
    \begin{array}{lcl}
      R_U   & = & R_U^0 \cup R_U^1, \\
      R_U^0 & = & \{\;\{u_x^0\} \to \emptyset \mid \mathbf{1}, \;
                  \{u_x^1\} \to \emptyset \mid \mathbf{1}, \;
                  \{u_y^0\} \to \emptyset \mid \mathbf{1}, \;
                  \{u_y^1\} \to \emptyset \mid \mathbf{1} \;\}, \\
      R_U^1 & = & \{\;\emptyset \to \{u_x^0\} \mid \mathbf{1}, \;
                  \emptyset \to \{u_x^1\} \mid \mathbf{1}, \;
                  \emptyset \to \{u_y^0\} \mid \mathbf{1}, \;
                  \emptyset \to \{u_y^1\} \mid \mathbf{1} \;\}. \\
    \end{array}
  \]

  Suppose that $F_U$ runs under the synchronous mode.  This is
  translated into the quasimode $\tilde M_{syn} = \{R\}$ for the
  Boolean P system $\Pi$.  The quasimode $\tilde M_U$ for $\Pi_U$ will
  be as follows:
  \[
    \tilde M_U = \{ R_U^0 \cup \tilde m^1_U \mid \tilde m^1_U \subseteq R_U^1 \}.
  \]
  Finally, the composite P system $\Pi \cup \Pi_U$ will run under the
  following quasimode:
  \[
    \tilde M \dottimes \tilde M_U = \{ R \cup R_U^0 \cup \tilde m^1_U \mid \tilde m^1_U \subseteq R_U^1 \}.
  \]

  \paragraph{Simulation of sequential control.}
  The 3 controls introduced in Example~\ref{ex:bcn} can be written as
  sets in the following way:
  \[
    \renewcommand{\arraystretch}{1.4}
    \begin{array}{lcl}
      \mu_1 &=& \{u_x^0, u_x^1, u_y^0, u_y^1\}, \\
      \mu_2 &=& \{\phantom{u_x^0,}\, u_x^1, u_y^0, u_y^1\}, \\
      \mu_3 &=& \{u_x^0, u_x^1, u_y^0 \phantom{, u_y^1}\}. \\
    \end{array}
  \]

  The trajectory $\tau_1 : 01 \to 10 \to 01$ of $F_U(\mu_1)$ will be
  simulated as the following evolution of $\Pi \cup \Pi_U$:
  \[
    \{y\} \cup \mu_1 \to \{x\} \cup \mu_1 \to \{y\} \cup \mu_1,
  \]
  where the rules to be applied in each transition are picked from the
  set $R \cup R_U^0 \cup R_U^1 \in \tilde M \dottimes \tilde M_U$.
  Note how $\mu_1$ is explicitly included as a set of symbols in the
  configuration of the composite Boolean P system $\Pi \cup \Pi_U$.

  Similary, the trajectory $\tau_2 : 01 \to 00 \to 00$ of $F_U(\mu_2)$
  will be simulated as follows:
  \[
    \{y\} \cup \mu_2 \to \emptyset \cup \mu_2 \to \emptyset \cup \mu_2,
  \]
  where the rules to be applied in each transition are picked from the
  set
  $R \cup R_U^0 \cup \{\, \emptyset \to \{u\} \mid \mathbf{1} \, \mid u \in \mu_2\} \in
  \tilde M \dottimes \tilde M_U$.  Note how all symbols corresponding
  to control inputs are removed at every step, and then specifically
  the control inputs from $\mu_2$ are reintroduced.

  Finally, the trajectory $\tau_3 : 00 \to 01 \to 11$ of $F_U(\mu_3)$
  will be simulated as follows by $\Pi \cup \Pi_U$:
  \[
    \emptyset \cup \mu_3 \to \{y\} \cup \mu_3 \to \{x,y\} \cup \mu_3.
  \]

  To simulate the final trajectory under the control sequence
  $\mu_{[3]} = (\mu_1, \mu_2, \mu_3)$, we glue together the final and
  the initial states of the above simulations, always anticipating the
  control from the subsequent simulation:
  \[
    \{y\} \cup \mu_1 \to \{x\} \cup \mu_1
    \to
    \underline{\{y\} \cup \mu_2} \to \emptyset \cup \mu_2
    \to
    \underline{\emptyset \cup \mu_3} \to \{y\} \cup \mu_3 \to \{x,y\} \cup \mu_3.
  \]
  Underlined elements are the states in which the control inputs
  change.  Thus, the transition
  $\{x\} \cup \mu_1 \to \underline{\{y\} \cup \mu_2}$ for example is
  governed by the set of rules
  $R \cup R_U^0 \cup \{\, \emptyset \to \{u\} \mid \mathbf{1} \, \mid u \in \mu_2\} \in
  \tilde M \dottimes \tilde M_U$ already, instead of
  $R \cup R_U^0 \cup R_U^1$ which was used in the first step.  \qed
\end{example}

The component $\Pi_U$ in the composite P system of
Theorem~\ref{thm:bp-seq} and Example~\ref{ex:pi-pi-u}
is an explicit implementation of the master
dynamical system driving the evolution of the controlled system $\Pi$.
The setting of this theorem captures the situation in which the
control can change at any moment, but $\Pi_U$ can be designed to
implement other kinds of control sequences.  We give the construction
ideas for the kinds of sequences introduced in~\cite{PardoID21}:
\begin{itemize}
\item \emph{Total Control Sequence (TCS):} all controllable variables
  are controlled at all times.

  \smallskip

  The quasimode of $\Pi_U$ will be correspondingly defined to always
  freeze the controlled variables:
  $\tilde M_U = \{R_U^0\} \dottimes 2^{P_U^1}$, where
  $P_U^1 \subseteq R_U^1$ with the property that for every $x_i \in X$
  every set $p \in P_U^1$ either introduces $u_i^0$ or $u_i^1$, but
  not both.

  \medskip

\item \emph{Abiding Control Sequence (ACS):} once controlled,
  a variable stays controlled forever, but the value to which it is
  controlled may change.

  \smallskip

  The rules of $\Pi_U$ will be constructed to never erase the control
  symbols which have already been introduced, but will be allowed to
  change the value to which the corresponding controlled variable will
  be frozen: $R_U = R_U^1 \cup P_U$, with the new set of rules defined
  as follows:
  \[
    P_U = \left\{\;\{u_i^a\} \to \{u_i^b\} \mid \mathbf{1} \, \mid x_i
      \in X, \, a, b \in \{0,1\}\right\}.
  \]
  $\Pi_U$ will able to rewrite some of the control symbols, or to
  introduce new control symbols: $\tilde M_U = 2^{R_U}$.
\end{itemize}

\section{Reachability in Boolean P Systems}
\label{sec:boolp-reach}

In this section we focus on reachability in Boolean P systems, which
we define in the following way: given a Boolean P system $\Pi$, a mode
$M$ (or a quasimode $\tilde{M}$), a set of starting states $S_\alpha$
and a set of target states $S_\omega$, decide whether an evolution of
$\Pi$ exists under the mode $M$ (or the quasimode $\tilde{M}$) driving
it from each state in $S_\alpha$ to some state in $S_\omega$.  We refer
to such a decision problem by the 4-tuple
$(\Pi, M^\dagger, S_\alpha, S_\omega)$, where $\mathcal M^\dagger$ may be
a mode or a quasimode.  In the rest of the paper, we will mainly deal
with reachability under quasimodes.

\begin{remark}
  Unlike the CoFaSe problem in which the synchronous mode is
  implicitly assumed, we explicitly include here the mode or the
  quasimode into the reachability problem.  Indeed, the size of the
  quasimode may be as much as exponential in the number of symbols,
  while the complexity of a mode may be even bigger, since it depends
  on the current configuration.  Furthermore, the mode choice impacts
  the answer of the problem.  For example the problem under the
  quasimode $\tilde{M} = \emptyset$ has a solution if and only if
  $S_\alpha \subseteq S_\omega$.
\end{remark}

In this section we will show that the reachability problem for Boolean
P~systems is \PSPACE-complete.  We start by showing that this
reachability problem is at least as hard as \LBAACCEPTANCE.

\begin{lemma}\label{lem:BPS_sim_LBA}
  $\LBAACCEPTANCE$ is reducible in polynomial time to reachability for
  Boolean P systems.
\end{lemma}

\begin{proof}
  We will first show how to construct a Boolean P system simulating
  a given LBA, and will then evaluate the size complexity of
  the construction.

  \paragraph{Construction.}
  Let $\mathcal{M} = (Q,V,T_1, T_2, \delta, q_0, q_1, Z_l, B, Z_r)$ be
  an LBA and $x \in T_1^*$ an input word of length $n$.  We construct
  in polynomial time a Boolean P system $\Pi = (\tilde{V},R)$ that
  simulates the computation of $\mathcal{M}$ on the input $x$.
  The alphabet of $\Pi$ contains the following symbols
  \[
    \tilde{V} = \lbrace A_{v,j} \mid v \in V, \; 0 \leq j
    \leq n+1 \rbrace \cup \lbrace C_{q,j} \mid q \in Q, \; 0 \leq j
    \leq n+1 \rbrace,
  \]
  where the symbols $A_{v,j}$ describe which symbols appear in which
  tape cells of
  $\mathcal{M}$ and $C_{q,j}$ describes the position and the state of
  the LBA head.  More precisely:
  \begin{itemize}
  \item $A_{v,j}$ represents the situation in which cell $j$ contains
    the symbol $v$,
  \item $C_{q,j}$ represents the situation in which the head is on
    cell $j$ and in state $q$.
  \end{itemize}

  We construct the rules of $\Pi$ as the union
  $R = \bigcup_{\rho \in \delta} R_{\rho}$, where each instruction
  $\rho = (q,X;p,Y,D)$ of $\mathcal{M}$ is simulated by a set of
  Boolean P system rules in the following way, depending on the
  direction of the movement of the head:
  \[
    \begin{array}{lll}
      D = R: & R_{(q,X;p,Y,R)} = \{ \, \lbrace A_{X,j}, C_{q,j} \rbrace
               \rightarrow \lbrace A_{Y,j}, C_{p,j+1} \rbrace \mid \mathbf{1}
               \,& \mid 0\leq j \leq n \, \}, \\
      D = S: & R_{(q,X;p,Y,S)} = \{ \, \lbrace A_{X,j}, C_{q,j} \rbrace
               \rightarrow \lbrace A_{Y,j}, C_{p,j} \rbrace \mid \mathbf{1}
               \,& \mid 0\leq j \leq n+1 \, \}, \\
      D = L: & R_{(q,X;p,Y,L)} = \{ \, \lbrace A_{X,j}, C_{q,j} \rbrace
               \rightarrow \lbrace A_{Y,j}, C_{p,j-1} \rbrace \mid \mathbf{1}
               \,& \mid 1\leq j \leq n+1 \, \}.
    \end{array}
  \]

  The evolution of $\Pi$ is governed by the quasimode
  $\tilde M = \{R\}$.  Due the form of the left-hand sides of the
  rules above, if the current state contains exactly one state symbol
  of the form $C_{q,j}$, at most one rule in $R$ will be applicable.

  We finally define the singleton set of target states:
  \[
    S_\omega = \lbrace \lbrace A_{B,j} \; \mid 1 \leq j \leq n \rbrace
    \cup \lbrace C_{q_1,0}, A_{Z_l,0}, A_{Z_r, n+1} \rbrace \rbrace.
  \]
  The only state appearing in $S_\omega$ therefore corresponds to the
  halting configuration of $\mathcal{M}$ in which all tape cells are
  blank except cells 0 and $n+1$ which contain the left and right end
  delimiters $Z_l$ and $Z_r$ respectively, and the head is on cell
  0 and in state $q_1$.

  It is a direct consequence of the definition of the rules in $R$
  that the LBA $\mathcal{M}$ accepts a word $x = v_1 v_2 \dots v_n$ of
  length $n$ if and only if the reachability problem
  $(\Pi, \tilde{M}, \lbrace s_x \rbrace, S_\omega )$ has a solution,
  where
  $s_x = \lbrace A_{v_j,j} \mid 1 \leq j \leq n \rbrace \cup \lbrace
  C_{q_0,1} \rbrace $.

  \paragraph{Complexity.}
  The number of symbols in $\Pi$ is $|\tilde{V}| = (n+2) (|V| + |Q|)$
  and the number of rules is $|R| = \mathcal{O}(n |V| |Q|)$, so the
  Boolean P system $\Pi$ can be constructed in time
  $\mathcal{O}(n |V| |Q|)$.

  Since $\tilde{M}$ is a singleton and its only element is of cardinal
  $|R| = \mathcal{O}(n|V| |Q|)$, the quasimode can be constructed in
  time
  $\mathcal{O}\left(n|V| |Q| \cdot \log(n|V| |Q|)\right)$---roughly,
  the number of rules times the number of bits necessary to
  describe a rule. Because there is only one starting state and one
  target state, and since a state can be described by a sequence of
  $n+3$ symbols ($n+2$ for the tape and $1$ for the state of the
  head), the whole description $(\Pi, \tilde{M}, S_\alpha, S_\omega)$
  can be constructed in the following time:
  \[
    \mathcal{O}\left(n |V| |Q| \cdot \log(n |V| |Q|)\right) =
    \mathcal{O}\left( (n |V| |Q|)^2\right).
  \]
  This expression is polynomial in the size of the specification of
  $\mathcal{M}$ and in the length $n$ of the input $x$, which
  concludes the proof. \qed
\end{proof}

We will now show the symmetrical statement that reachability in
Boolean P~systems is at most as hard as \LBAACCEPTANCE.

\begin{lemma}\label{lem:BPS_in_PSPACE}
  Reachability for Boolean P systems is in \PSPACE.
\end{lemma}

\begin{proof}
  We will prove that reachability for Boolean P systems is in
  \NPSPACE, which implies the required statement by Savitch's theorem
  \cite{savitch1970relationships}.

  \smallskip

  Let $(\Pi, \tilde{M}, S_\alpha, S_\omega)$, with $\Pi = (V, R)$, be
  an instance of the reachability problem.
  Algorithm \ref{alg:BPS-reach} is a non-deterministic algorithm that
  solves this problem in polynomial space.
  \begin{algorithm}
    \caption{Solving reachability for Boolean P systems in \NPSPACE}
    \label{alg:BPS-reach}
    \begin{algorithmic}
    \Require $(\Pi, \tilde{M}, S_\alpha, S_\omega)$, $\Pi = (V, R)$
    \Ensure $\textit{Reachable} = \textit{true} \iff (\Pi, \tilde{M}, S_\alpha, S_\omega) \text{ has a solution} $
    \ForAll{$x \in S_\alpha$}
      \State $i \gets 0$
      \State $s \gets x$
      \State $\textit{Reachable}_x \gets \textit{false}$
      \While{$i < 2^{|V|}$}
        \State $i \gets i + 1$
        \If{$s \in S_\omega$}
          \State $\textit{Reachable}_x \gets \textit{true}$
        \EndIf
        \State Pick non-deterministically $m \in \tilde{M}$
        \State $s \gets \textit{UPDATE}_\Pi(s, m)$
      \EndWhile
    \EndFor
    \State $\textit{Reachable} \gets \bigwedge_{x \in S_\alpha} \textit{Reachable}_x$
  \end{algorithmic}
  \end{algorithm}
  The function $\textit{UPDATE}_\Pi$ takes a state $s$ of $\Pi$ and an
  element of a quasimode $m \in \tilde{M}$, and returns the state
  updated according to the rules $R$ of $\Pi$ and the chosen element
  of the quasimode, as defined in Section~\ref{sec:quasimodes}.

  Since the number of possible states of $\Pi$ is $2^{|V|}$, the
  shortest evolution between two states is of length at most
  $2^{|V|}$, if it exists. Algorithm \ref{alg:BPS-reach} therefore
  non-deterministically tests all possible evolutions of length at
  most $2^{|V|}$, starting from all states in $S_\alpha$.  At the end
  \textit{Reachable} gets the value \textit{true} if and only if
  a state in $S_\omega$ can be reached from every state in $S_\alpha$,
  which ensures the correctness of the algorithm.

  \smallskip

  This algorithm runs in polynomial space in the size of the
  reachability problem. Note that several states, a counter up to
  $2^{|V|}$, and $|S_\alpha|$ Boolean flags are stored, all of which
  takes up $\mathcal{O}(|V| + |S_\alpha|)$ space. Furthermore, the
  function $\textit{UPDATE}_\Pi$ can be evaluated in polynomial
  space. Indeed, to determine the set of applicable rules in a state
  $s$, one needs to check for each rule if the guard is true and if
  the left part of the rule is present in $s$. Both operations, the
  evaluation of a Boolean function and a comparison, can be carried
  out in polynomial space with respect to $|V|$. Only the rules in $m$
  are then applied, and these applications can be carried out in
  polynomial space with respect to $|V|$ and $|R|$. \qed
\end{proof}

\begin{remark}
  The argument of Lemma~\ref{lem:BPS_in_PSPACE} focuses on
  reachability under quasimodes.  This argument can be trivially
  extended to modes derivable from quasimodes, and more generally to
  any mode for which non-deterministically picking a set $m$ of rules
  to apply can be done in polynomial space.
\end{remark}

The following theorem brings together Lemmas~\ref{lem:BPS_sim_LBA}
and~\ref{lem:BPS_in_PSPACE} to show the main result with respect to
the complexity of reachability.

\begin{theorem}\label{thm:BooleanPsys_PSPACE}
  Reachability for Boolean P systems is \PSPACE-complete
\end{theorem}

\section{Complexity of Sequential Controllability}
\label{sec:seq-pspace}

In this section we first extend the CoFaSe problem with some
additional details necessary to properly reason about its complexity,
and then show that sequential controllability of BCN is
\PSPACE-complete.

\subsection{CoFaSe and Control Modes}
Theorem~\ref{thm:bp-seq} shows that Boolean P systems can directly
simulate Boolean networks together with the master control system, and
Theorem~\ref{thm:BooleanPsys_PSPACE} shows that reachability for
Boolean P systems is \PSPACE-complete.  Nevertheless, we cannot
immediately conclude that CoFaSe is \PSPACE-complete because of the
role modes and quasimodes play in evaluating the size of the problem.

Consider a BCN $F_U$ with the variables $X$ and the control inputs
$U$, and recall that the CoFaSe problem is given by the triple
$(F_U, S_\alpha, S_\omega)$, where $S_\alpha, S_\omega \subseteq S_X$
are the sets of starting and target states respectively.
The simulating Boolean P~system $\Pi \cup \Pi_U$ constructed in
Theorem~\ref{thm:bp-seq} uses the quasimode
\[
  \tilde M_{U} =  \{R_U^0\} \dottimes 2^{R_U^1},
\]
for which $|\tilde M_U| = 2^{|U|}$, meaning that size of the
reachability problem for $\Pi \cup \Pi_U$ is always exponential in the
size of $U$, independently of the sizes of the individual elements of
the triple $(F_U, S_\alpha, S_\omega)$ \footnote{In general, the
  description of $F_U$ is of size $\mathcal{O}(2^{|X||U|})$, because
  some Boolean functions may require an exponential number of Boolean
  connectors $\wedge$, $\vee$, $\bar \cdot$ to be represented.
  $S_\alpha$ and $S_\omega$ are of size $\mathcal{O}(|X|)$ by their
  definition.  In practice, however, the sizes of these entities are
  often well under the respective upper
  bounds~\cite{BianeD19,PardoID21}.}.  As a consequence, directly
combining Theorems~\ref{thm:bp-seq} and~\ref{thm:BooleanPsys_PSPACE}
is not guaranteed to yield a polynomial bound on space in terms of the
size of the CoFaSe problem $(F_U, S_\alpha, S_\omega)$.


We believe that the correct way to deal with this issue is to
include a specification of the master system emitting the
controls into the description of the problem of sequential
controllability.  Indeed, CoFaSe is formulated for the situation in
which the control can change at any moment~\cite{PardoID21}, and this
information is not explicitly included in its definition, while it is
explicitly present in the P system $\Pi \cup \Pi_U$ from
Theorem~\ref{thm:bp-seq}.

We propose to describe the possible changes in controls by defining
a relation on $2^U$---the control mode.  A \emph{control mode} for
a BCN $F_U$ is a relation $\mathcal{R}_U \subseteq 2^U \times 2^U$
describing the possible evolutions of control inputs.  More precisely,
consider the following trajectory of the BCN $F_U$:
\[
  s_1 \xrightarrow{F_U(\mu_1)} s_2 \xrightarrow{F_U(\mu_2)} s_3
  \xrightarrow{F_U(\mu_3)} \dots \xrightarrow{F_U(\mu_n)} s_{n+1}.
\]
This trajectory complies with the control mode $\mathcal{R}_U$ if and
only if $(\mu_i, \mu_{i+1}) \in \mathcal{R}_U$, for every
$1 \leq i \leq n$.

\begin{example}
  Control modes naturally capture the types of control sequences given
  at the end of Section~\ref{sec:boolp-seq} and initially discussed
  in~\cite{PardoID19}.  To streamline the definitions of the
  corresponding control modes, we introduce the following helper
  function:
  \[
    \textit{idx} : 2^U \to 2^{\{1, \dots, |U|\}}, \quad
    \textit{idx}(\mu) = \{i \mid u_i^\star \in \mu, \star \in \{0,
    1\}\}.
  \]
  In other words, \textit{idx} produces the set of control input
  indices appearing in a control $\mu$, irrespectively of the nature
  of the control input (freeze to 0 or freeze to 1).

  We can now define the control mode $\mathcal{R}_U^\textit{TCS}$
  capturing Total Control Sequences (TCS) as follows:
  \[
    \forall \mu,\nu \in 2^U :
    (\mu,\nu) \in \mathcal{R}_U^\textit{TCS}
    \iff
    \textit{idx}(\mu) = \textit{idx}(\nu) = \textit{idx}(U).
  \]
  Informally, $\mathcal{R}_U^\textit{TCS}$ includes all those pairs of
  controls which act on every single controlled variable by activating
  one of the corresponding control inputs.

  Similarly, in the case of Abiding Control Sequences (ACS), the
  control mode $\mathcal{R}_U^\textit{ACS}$ can be defined as follows:
  \[
    \forall \mu,\nu \in 2^U : (\mu,\nu) \in \mathcal{R}_U^\textit{TCS}
    \iff \textit{idx}(\mu) \subseteq \textit{idx}(\nu).
  \]
  Thus, $\mathcal{R}_U^\textit{ACS}$ only allows to transition from
  $\mu$ to $\nu$ if $\nu$ acts at least on the same controlled
  variables as $\mu$.  Note that $\nu$ is allowed to change the value
  to which a controlled variable $x_i$ is frozen by replacing $u_i^0$
  by $u_i^1$ or vice versa.  \qed
\end{example}

We now define an extension of CoFaSe to capture sequential
controllability of BCN in a more general framework.  The $\SEQCONTROL$
problem is given by the 5-tuple
$(F_U, M, \mathcal{R}_U, S_\alpha, S_\omega)$ and consists in deciding
whether for every starting state in $S_\alpha$ there exists an initial
control $\mu_0$ and a trajectory of the BCN $F_U$ under the mode $M$
and the control mode $\mathcal{R}_U$ ending up in a target state from
$S_\omega$.  $\mu_0$ must appear as a the first term in at least
a pair of $\mathcal{R}_U$:
$\exists \nu \subseteq U : (\mu_0, \nu) \in \mathcal{R}_U$.

\begin{example}\label{ex:seq-control-mode}\nocite{eMathHelp,dds}
  Consider the Boolean network $F_U$ described in
  Figure~\ref{fig:seq-control-mode}, as well as the controls
  $\mu_{110} = \{u_1^1, u_2^1\}$, freezing both $x_1$ and $x_2$ to 1,
  and $\mu_\emptyset = \emptyset$.  If
  $(\mu_{110}, \mu_{110}), (\mu_{110}, \mu_\emptyset) \in
  \mathcal{R}_U$ then the following trajectory is possible:
  \[
    000 \xrightarrow{F_U(\mu_{110})}
    110 \xrightarrow{F_U(\mu_{110})}
    111 \xrightarrow{F_U(\mu_\emptyset)}
    001.
  \]

  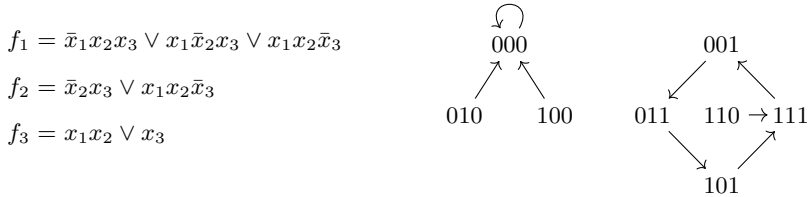
\begin{figure}
    \centering%
    {\renewcommand{\arraystretch}{1.6}
      $\begin{array}{lcl}
         f_1 &=& \bar x_1 x_2 x_3  \vee  x_1 \bar x_2 x_3  \vee  x_1 x_2 \bar x_3 \\
         f_2 &=& \bar x_2 x_3  \vee  x_1 x_2 \bar x_3 \\
         f_3 &=&  x_1 x_2 \vee x_3
       \end{array}$
     }
     \hspace{9mm}
     \begin{tikzpicture}[node distance=5mm and 2mm,baseline,yshift=6.1mm]
       \node (000) {$000$};
       \node[below=of 000,xshift=-6mm] (010) {$010$};
       \node[below=of 000,xshift=6mm] (100) {$100$};

       \node[right=21mm of 000] (001) {$001$};
       \node[below=of 001] (110) {$110$};
       \node[left=of 110] (011) {$011$};
       \node[right=of 110] (111) {$111$};
       \node[below=of 110] (101) {$101$};

       \draw[->] (000) to[out=60,in=120,looseness=5] (000);
       \draw[->] (001) -- (011);
       \draw[->] (010) -- (000);
       \draw[->] (011) -- (101);
       \draw[->] (100) -- (000);
       \draw[->] (101) -- (111);
       \draw[->,shorten >=-.5mm] (110) -- (111);
       \draw[->] (111) -- (001);
     \end{tikzpicture}
     \caption{The update functions of the Boolean network from
       Example~\ref{ex:seq-control-mode} (left) as well as its
       uncontrolled synchronous dynamics (right).}
    \label{fig:seq-control-mode}
  \end{figure}

  Suppose now that only freezing $x_1$ or $x_2$ separately is
  permitted, i.e.
  $i \in \textit{idx}(\mu) \cap \{1, 2\} \implies
  \textit{idx}(\mu) = \{i\}$, for any $\mu$ appearing in a pair in
  $\mathcal{R}_U$.  In this case, $F_U$ can reach $100$ or $010$ from
  $000$ by respectively controlling $x_1$ or $x_2$ to~1.
  The following 3 scenarios are possible afterwards:
  \begin{enumerate}
  \item maintain the control of $x_1$ or $x_2$ and stay in the same
    state;
  \item freeze the other variable---$x_2$ if $x_1$ was controlled or
    $x_1$ if $x_2$ was controlled, and switch to the other
    state---$010$ or $100$ respectively;
  \item release all controls and go back to $000$.
  \end{enumerate}
  In any of these cases, $F_U$ is not able to reach $001$ with the
  above restriction on the control mode.

  Finally, suppose that once $\mu_{110}$ is employed, it must be
  maintained for the rest of the trajectory, i.e.
  $(\mu_{110}, \nu) \in \mathcal{R}_U \implies \nu = \mu_{110}$.
  The previous paragraph shows that the only way for $F_U$ to leave
  the connected component consisting of the states $\{000, 100, 010\}$
  while starting from $000$ is to apply $\mu_{110}$.  On the other
  hand, since $\mu_{110}$ cannot be deactivated once applied, this
  means that $F_U$ cannot reach $001$ from $000$ with this restriction
  on the control mode. \qed
\end{example}

\subsection{$\SEQCONTROL$ and CoFaSe Are \PSPACE-complete}
\label{sec:pspace-complete}

We start by combining Theorems~\ref{thm:bp-seq}
and~\ref{thm:BooleanPsys_PSPACE} to characterize the complexity of
$\SEQCONTROL$.

\begin{theorem}
  \label{thm:seq-bcn-PSPACE}
  $\SEQCONTROL$ is \PSPACE-complete.
\end{theorem}

\begin{proof}
  $\SEQCONTROL$ is \PSPACE-hard, since by taking $U = \emptyset$ it is
  reduced to the problem of reachability for Boolean networks, known
  to be \PSPACE-complete~\cite{ChengEP1995,PardoID19}.

  Let now $(F_U, M, \mathcal{R}_U, S_\alpha, S_\omega)$ be an instance
  of $\SEQCONTROL$ and consider the following set of rules:
  \[
    R_U = \{\, \mu_1 \to \mu_2 \mid \mathbf{1} \, \mid (\mu_1, \mu_2) \in
    \mathcal{R}_U \,\},
  \]
  as well as the quasimode $\tilde M_U = \{r \mid r \in R_U\}$.
  The Boolean P system $\Pi_U = (U, R_U)$ running under the quasimode
  $\tilde M_U$ will therefore simulate the changes in controls allowed
  by the control mode $\mathcal{R}_U$.

  We can now construct the reachability problem
  $(\Pi \cup \Pi_U, \tilde{M} \dot{\times} \tilde{M}_U, S_\alpha,
  S_\omega)$ in the same way as in Theorem \ref{thm:bp-seq}.
  The entire construction, including that of $R_U$, happens in
  polynomial time with respect to the size of the initial instance of
  $\SEQCONTROL$.  This allows us to conclude the proof by invoking the
  fact that reachability in Boolean P systems is \PSPACE-complete
  (Theorem~\ref{thm:BooleanPsys_PSPACE}). \qed
\end{proof}

As explained in the previous section, $\SEQCONTROL$ being
\PSPACE-complete does not immediately imply that CoFaSe is
\PSPACE-complete, since translating from CoFaSe to $\SEQCONTROL$ may
require exponential increase in space.  However, it is possible to
directly prove that CoFaSe is in $\PSPACE$ by using a variation of
Algorithm~\ref{alg:BPS-reach} from Lemma~\ref{lem:BPS_in_PSPACE}.

\begin{theorem}\label{thm:CoFaSe-PSPACE}
  CoFaSe is PSPACE-complete.
\end{theorem}

\begin{proof}
  Similarly to the proof of Lemma~\ref{lem:BPS_in_PSPACE}, we show
  here a non-deterministic polynomial-space algorithm solving the
  instance of CoFaSe given by the triple $(F_U,S_\alpha, S_\omega)$:
  Algorithm~\ref{alg:CoFaSe}.

  \begin{algorithm}
    \caption{Solving CoFaSe in \NPSPACE}\label{alg:CoFaSe}
    \begin{algorithmic}
    \Require $(F_U,S_\alpha, S_\omega)$
    \Ensure $Reachable = true \iff (F_U, S_\alpha, S_\omega) \text{ has a solution} $
    \ForAll{$x \in S_\alpha$}
      \State $i \gets 0$
      \State $s \gets x$
      \State $\textit{Reachable}_x \gets \textit{false}$
      \While{$i < 2^{|X|}$}
        \State $i \gets i + 1$
        \If{$s \in S_\omega$}
          \State $\textit{Reachable}_x \gets \textit{true}$
        \EndIf
        \State Pick non-deterministically $\mu \in S_U$
        \State $s \gets F_U(\mu)(s)$
      \EndWhile
    \EndFor
    \State $\textit{Reachable} \gets \bigwedge_{x \in S_\alpha} \textit{Reachable}_x$
    \end{algorithmic}
  \end{algorithm}

  Algorithm~\ref{alg:CoFaSe} has very similar properties to
  Algorithm~\ref{alg:BPS-reach}.  Note that no requirements on the
  values of control inputs are imposed in the CoFaSe problem, meaning
  that only the state space $S_X$ needs to be explored, excluding the
  control inputs.  Since $|S_X| = 2^{|X|}$, exploring trajectories of
  length at most $2^{|X|}$ is sufficient to conclude about the
  reachability of a state in $S_\omega$ for all states in $S_\alpha$.

  Algorithm~\ref{alg:CoFaSe} stores a constant number of intermediate
  states and controls, a counter up to $2^{|X|}$, and $|S_\alpha|$
  Boolean flags, all of which takes up
  $\mathcal{O}(|V| + |U| + |S_\alpha|)$ space.  Furthermore, $F_U$ can
  be computed in polynomial space in $|X|$ and $|U|$, meaning that
  Algorithm~\ref{alg:CoFaSe} requires polynomial space in the size of
  the triple $(F_U,S_\alpha, S_\omega)$.  Finally, we conclude the
  proof by invoking Savitch's theorem \cite{savitch1970relationships},
  stating that $\NPSPACE = \PSPACE$. \qed
\end{proof}

\section{Conclusion and Discussion}
\label{sec:conclusion}

We structure the conclusion into three subsections, focusing on three
main takeaways and future research directions stemming from
this paper.

\subsection{Complexity of Sequential Controllability}
The central technical result of this work is proving that sequential
controllability of Boolean control networks (BCN) is \PSPACE-complete,
thereby closing the question left open in~\cite{PardoID21}.
One important intuition that this result yields is that sequential
controllability of BCN is not in fact harder computationally speaking
than simple reachability, in spite of the much heftier two-level setup
with a master dynamical system driving the Boolean network.  While no
explicit construction is given, it is to be expected that the
evolution of a BCN under a control sequence may be simulated by
a Boolean network, modulo a polynomial transformation.  This implies that
reasoning about sequential controllability is as hard as reasoning
about pure reachability in Boolean networks, opening a promising
direction of future work about using the most permissive
semantics~\cite{PauleveKCH2020} for sequential controllability of BCN.

We stress nevertheless that sequential controllability and
reachability being in the same complexity class does not necessary
imply that the techniques for efficiently solving reachability in
practical situations can be immediately transposed to controllability.
Exploring such possibilities is an important direction for future
research on sequential controllability of BCN.

While we extensively deal with CoFaSe in this work, it should be noted
that the ConEvs semantics explored in~\cite{PardoID21} is not treated.
The ConEvs semantics of the control sequence constraints the moments
at which the control may change to the stable states of the driven
Boolean network.  This places the master system in a feedback loop
with the driven network and changes the architecture substantially.
In particular, the computational complexity of sequential
controllability under the ConEvs semantics still remains to
be characterized.

\subsection{Boolean P Systems}
Most of the technical results presented in this paper are obtained via
Boolean P systems, a framework specifically designed for dealing with
sequential controllability in Boolean networks.  We particularly
emphasize one of our central goals: designing ad hoc formalisms very
tightly suited for a specific problem and thereby giving new
relevant viewpoints.

One of the advantages in relying on Boolean P systems is that the
language of individual rules is more flexible than that of
propositional formulae in Boolean networks.  In particular, having set
rewriting directly available allows for naturally expressing the
notions of adding, removing, or depending on resources, while the
propositional guards allow for easy checking of Boolean conditions
whenever necessary.  These two ingredients shine in
Section~\ref{sec:boolp-seq}, in which we show how a Boolean P system
can capture both the BCN and the master dynamical system emitting the
control inputs.  On the other hand, we construct Boolean P~systems
without indulging too much into computationally expensive ingredients,
which keeps the complexity of reachability in \PSPACE.

We would like to dwell specifically on the difference between
Theorems~\ref{thm:seq-bcn-PSPACE} and~\ref{thm:CoFaSe-PSPACE}, in
particular on the fact that the latter directly shows that CoFaSe is
in \PSPACE, completely eliding Boolean P systems.  First, remark that
Theorem~\ref{thm:seq-bcn-PSPACE} showing that $\SEQCONTROL$ is
$\PSPACE$-complete is in fact more general, as it holds for any mode
and for any control mode, incorporating \emph{en passant} different
kinds of control sequences such as TCS, ACS, etc.  Secondly, remark
that Algorithm~\ref{alg:CoFaSe} in Theorem~\ref{thm:seq-bcn-PSPACE} is
directly derived from (and is a special case of)
Algorithm~\ref{alg:BPS-reach} in Lemma~\ref{lem:BPS_in_PSPACE}, which
arguably needed some general framework like Boolean P systems to
be conceived.

Going back to the ConEvs semantics mentioned in the previous
subsection, we expect that considering it in the framework of Boolean
P systems will bring new valuable insight both concerning the
characterization of its complexity and its other properties, as well
as possible optimizations for specific use cases.  Observe that ConEvs
cannot be captured as a control mode, because it introduces a backward
dependency of the control sequence on the state of the BCN.
Boolean P systems on the other hand should allow to express this
feedback elegantly, since the master system $\Pi_U$ and the driven
system $\Pi$ are both part of the same composite system
$\Pi \cup \Pi_U$ (Theorem~\ref{thm:bp-seq}), and can therefore
communicate both ways.  In fact, just from this informal reasoning we
can make a conjecture with respect to the upper bound on the
complexity of sequential controllability under the ConEvs semantics.

\begin{conjecture}
  Sequential controllability of BCN under the ConEvs semantics is in
  \PSPACE.
\end{conjecture}

Finally, we stress once again the point of Remark~\ref{rem:boolp-rs}:
while Boolean P systems are very closely related to reaction
systems~\cite{EhrenfeuchtR2007}, they have distinctive features which
make them a much better fit for reasoning about sequential
controllability---specifically, explicit Boolean guards and permanency
of the resources.

\subsection{Lineage of (Polymorphic) P Systems, Homoiconicity, and
  Lisp}
As we have already insisted, one central point that we bring forward
with this work is conceiving ad hoc formalisms specialized for solving
particular problems.  This approach is partially inspired by the venerable Lisp
family of programming languages, and more particularly by
language-oriented programming---a methodology proposing to start
solving problems by developing specifically-tailored programming
languages---domain-specific languages or
DSLs~\cite{FelleisenFFKBMT18,Ward1994}.

When adopting this approach, it is important that such bespoke
constructions be done within a particular general framework, lest the
design costs grow too high and the new formalisms too obscure.
In this paper, we promote P systems as such a general framework.
The community around this model of computing has been producing a wide
spectrum of variants, a far-from-exhaustive glimpse of which can be
seen in~\cite{AlhazovFI2016,imcs,JMC2022,handMC}.  The rich body of
literature provides many ingredients and various tools for easily
assembling different new formalisms.  This is why we believe that
P systems are particularly well suited for the ad hoc
formalism methodology.

We conclude this work by underlining that Boolean P systems are far
from being a frontier of how far one can go in designing specialized
formalisms.  We recall as an example polymorphic
P systems~\cite{AlhazovIR10}, in which the rules are given by pairs of
membranes rather than being part of the static description of the
system, as is classically done in automata and language theory.
Polymorphic P~systems thus implement a form of
homoiconicity---code-as-data, similarly to the Lisp languages.  A lot
more can be done in terms of customizing P systems, and we expect to
see and invest further effort into the ad hoc formalism methodology.

\subsection*{Acknowledgements}
All authors are grateful to Laurent Trilling from Université Grenoble
Alpes for fruitful discussions.

Artiom Alhazov acknowledges project 20.80009.5007.22 ``Intelligent
information systems for solving ill-structured problems, processing
knowledge and big data'' by the National Agency for Research
and Development.

\bibliographystyle{plain}
\bibliography{TCS2023PvsB}

\end{document}